\documentclass[letterpaper, 10 pt, conference]{ieeeconf}  

\IEEEoverridecommandlockouts                   

\usepackage[noadjust]{cite}
\usepackage{hyperref}
\usepackage{graphics}
\usepackage{epsfig} 
\usepackage{mathptmx}
\usepackage{times}
\usepackage{amsmath}
\usepackage{amssymb}
\usepackage{dsfont}
\usepackage{enumerate}
\usepackage[tight]{subfigure}
\usepackage{multirow}
\usepackage{bm}

\usepackage{color}

\title{
\bf Consensus on Open Multi-Agent Systems \linebreak Over Graphs Sampled from Graphons
}

\author{Renato Vizuete and Julien M. Hendrickx
\thanks{*This work was supported by F.R.S.-FNRS via the \emph{KORNET} project, and by the \emph{SIDDARTA} Concerted Research Action (ARC) 
of the Fédération Wallonie-Bruxelles.}
\thanks{R.~Vizuete and J.~M.~Hendrickx are with ICTEAM institute, UCLouvain, B-1348, Louvain-la-Neuve, Belgium. R.~Vizuete is a FNRS Postdoctoral Researcher - CR.
{\tt\small renato.vizueteharo@uclouvain.be},
{\tt\small julien.hendrickx@uclouvain.be}\protect.}}

\newcommand{\vertiii}[1]{{\left\vert\kern-0.25ex\left\vert\kern-0.25ex\left\vert #1 
    \right\vert\kern-0.25ex\right\vert\kern-0.25ex\right\vert}}

\newcommand{\ones}{\mathds{1}}
\newcommand{\SBM}{\text{SBM}}

\newtheorem{definition}{Definition}
\newtheorem{theorem}{Theorem}

\newtheorem{proposition}{Proposition}
\newtheorem{lemma}{Lemma}

\newcommand{\erf}{\text{erf}}

\newcommand{\nm}{n_{m}}
\newcommand{\nM}{n_{M}}

\newcommand{\prt}[1]{\left(#1\right)}
\newcommand{\brk}[1]{\left[#1\right]}

\newcommand{\abs}[1]{\left|#1\right|}
\newcommand{\norm}[1]{\abs{\abs{#1}}}

\newcommand{\R}{\mathbb R}
\newcommand{\N}{\mathbb N}

\newcommand{\Ep}[1]{\mathbb{E}\left[#1\right]}

\begin{document}

\maketitle
\thispagestyle{empty}

\begin{abstract}

We show how graphons can be used to model and analyze open multi-agent systems, which are multi-agent systems subject to arrivals and departures, in the specific case of linear consensus.
First, we analyze the case of replacements, where under the assumption of a deterministic interval between two replacements, we derive an upper bound for the disagreement in expectation. Then, we study the case of arrivals and departures, where we define a process for the evolution of the number of agents that guarantees a minimum and a maximum number of agents. Next, we derive an upper bound for the disagreement in expectation, and we establish a link with the spectrum of the expected graph used to generate the graph topologies. Finally, for stochastic block model (SBM) graphons, we prove that the computation of the spectrum of the expected graph can be performed based on a matrix whose dimension depends only on the graphon and it is independent of the number of agents. 
\end{abstract}

\section{Introduction}

Open multi-agent systems are a framework used to analyze networks subject to arrivals, departures or replacements of agents at a rate similar to the scale time of the process \cite{hendrickx2016open,hendrickx2017open}. This type of systems are essentially characterized by the agent internal dynamics, the evolution of the network and the arrivals and departures \cite{vizuete2024trends}. Due to the complexity of the system, most of the works focus mainly on the agent internal dynamics and the processes for arrivals and departures, 
neglecting the influence of the network dynamics (changes on the set of nodes and connections). This is usually done by considering \emph{trivial} dynamics like complete graphs \cite{monnoyer2021random,monnoyer2020open,vizuete2022resource}, bounds on the algebraic connectivity or diameter \cite{franceschelli2020stability,deplano2024stability} or just connectivity at all time instants \cite{deplano2025optimization}. Unfortunately, this fails to model more realistic open multi-agent systems since clearly the network dynamics can have a significant influence on the evolution of the states of the agents.

The analysis of the network dynamics is an extensive area of research and one of the main tools used to generate dense graphs with time-varying size are \emph{graphons} developed in \cite{lovasz2006limits,lovasz2012large}. Graphons could provide an important framework to analyze open multi-agent where the dynamics of the network play an important role, but, most of the theory of graphons has been focused on applications regarding the adjacency matrix \cite{gao2019graphon,parise2023graphon,vizuete2020graphon} and few works involve the Laplacian matrix \cite{vizuete2021laplacian,petit2021random,bramburger2023pattern}. Unfortunately, most of the dynamics over networks, including the  consensus-type dynamics, are functions of the Laplacian matrix, but due to the complexity of the Laplacian graphon operator, most of the works avoid using its spectrum for the analysis
\cite{kuehn2019power,bonnet2022consensus,prisant2024opinion}.  

We use here graphons to study open multi-agent systems over Laplacian matrix based dynamics. In our analysis, we will use an approach based on descriptors, which are scalar quantities associated with the dynamics of the system \cite{vizuete2024trends}. Depending on the type of application, the choice of the descriptor is natural and its objective is to distinguish the behavior of the system due to arrivals, departures or replacements. For instance, in consensus, the disagreement $
V(x)=\frac{1}{n}\sum_i (x_i-\bar x)^2
$ with $\bar x=\frac{1}{n}\sum_i x_i$ has been used in \cite{monnoyer2020open}, while in the case of epidemics, a Lyapunov function of the form $V(x)=\frac{1}{n}\norm{x}^2$ has been used in \cite{vizuete2024SIS}.

In this work, we analyze the linear consensus problem on open multi-agent systems when the graph topologies are sampled from a graphon. We use the disagreement as a descriptor and derive upper bounds for its asymptotic behavior in expectation. An analysis of consensus using the disagreement as a descriptor has been performed in \cite{monnoyer2020open}, where the authors study pairwise gossip interactions at discrete time instants in a complete graph. In our case, we analyze the linear consensus dynamics based on a Laplacian matrix, where the state of the system evolves in continuous time, and changes at discrete time instants corresponding to the potential events in open multi-agent systems (i.e., replacements, arrivals, departures). Furthermore, to the best of our knowledge, we consider for the first time, graph topologies sampled from graphons in open multi-agent systems, where the spectrum of the sampled graphs play an important role on the performance of the system and the derivation of upper bounds for its analysis.

\section{Problem formulation}

\subsection{Model and descriptor}

An undirected graph is defined as a pair $\mathcal{G}=(\mathcal{V},\mathcal{E})$ where $\mathcal{V}=\{1,\ldots,n\}$ is a finite set of vertices or nodes and 
$\mathcal{E}\subseteq  \mathcal{V} \times \mathcal{V}$ is the set of edges. 
The adjacency matrix of a graph $A=[a_{ij}]\in \mathbb{R}^{n\times n}$ is defined by $a_{ij}>0$ if $(i,j)\in E$ and $a_{ij}=0$ otherwise. The Laplacian matrix is defined as $L=D-A$, where $D=\text{diag}(d_1,\ldots,d_n)$ and each $d_i$ is the $i$-th row sum of $A$.

We consider a standard linear consensus, where the dynamics of each agent are given by
$$
\dot x_i=\sum_{j=1}^n a_{ij}(t)(x_j(t)-x_i(t)),
$$
and the dynamics of all the network can be expressed as 
\begin{equation}\label{eq:linear_consensus}
\dot{x}(t)=-L(t)x(t),    
\end{equation}
where $L(t)$ is the Laplacian matrix at the time instant $t$.
The agents follow a consensus like dynamics, 
and that is why we are going to use the disagreement as a descriptor \cite{vizuete2024trends} for the analysis in an open system:
$$
V(x(t))=\frac{1}{n}\norm{x(t)}^2-\bar x^2(t),
$$
where $\bar x^2=\prt{\frac{1}{n}\sum_{i=1}^n x_i}^2$ and $\norm{\cdot}$ is the Euclidean norm. When the system reaches consensus, the disagreement satisfies:
$$
\lim_{t\to\infty}V(x(t))=0,
$$
such that any deviation from zero is due exclusively to the impact of arrivals, departures or replacements in the system. In the rest of the paper, we will use $V(t)$ to lighten notation.

We consider that an arrival, departure or replacement, which we call \emph{event}, occur according to a discrete evolution of time where each time-step $k\in\N$ corresponds to the time instant at which the $k$-th event occurs. In addition, we assume that the Laplacian matrix $L$ can only change at a $k$-th event and remains constant between two time instants associated with events $k$ and $k+1$. Moreover, between the events $k$ and $k+1$, the system evolves in continuous time according to \eqref{eq:linear_consensus}. Notice that we use $t$ to denote the time and $k$ to denote the order of the sequence of events from $t=0$. Along this work, the notation $X(t^k)$ will be used to specify the value of $X$ at the time instant associated with the $k$-th event. During a departure, the agent that leaves the system is chosen uniformly among the current agents in the network. During an arrival, the value of the joining agent is taken from a continuous distribution with zero mean and variance $\sigma^2$. 

\begin{proposition}\label{prop:consensus_mappings}
If $x$ follows the linear consensus \linebreak dynamics \eqref{eq:linear_consensus},
the disagreement $V$ satisfies:
\begin{equation}\label{eq:consensus_continuous}    
 V(t^{k+1})\le V(t^k)e^{-2\lambda_2(t^k)\Delta t^{k,k+1}};
\end{equation}
\begin{equation}\label{eq:consensus_departure}
\Ep{V^+(t^k)|V(t^k),\text{Dep}_n}=\prt{1-\frac{1}{(n-1)^2}}V(t^k);
\end{equation}
\begin{equation}\label{eq:consensus_arrival}
\Ep{V^+(t^k)|V(t^k),\text{Arr}_n}\le \frac{n}{n+1}V(t^k)+\frac{\sigma^2}{n+1};
\end{equation}
\begin{equation}\label{eq:consensus_replacement}
\Ep{V^+(t^k)|V(t^k),\text{Rep}_n}\le \frac{n^2-n-1}{n^2}V(t^k)+\frac{n^2-1}{n^3}\sigma^2,
\end{equation}
where $V(t^k)$ and $V^+(t^k)$ denote the values of $V$ before and after the occurrence of event $k$ respectively, $\lambda_2(t^k)$ is the second smallest  Laplacian eigenvalue at the time instant $t^k$, and $\Delta t^{k,k+1}$ is the interval of time between  $t^k$ and $t^{k+1}$. 
\end{proposition}
\begin{proof}
Since the graph is symmetric, the function $V(t)$ satisfies during $\Delta t^{k,k+1}$:
\begin{align}
\dot{V}&=-\frac{2}{n}x^TLx\nonumber\\
&=-\frac{2}{n}\prt{x_\bot+\bar x\mathds{1}_n}^TL\prt{x_\bot+\bar x\mathds{1}_n}\nonumber\\
&\le -\frac{2\lambda_2}{n}\norm{x_\bot}^2\nonumber\\
&= -\frac{2\lambda_2}{n}\prt{\norm{x}^2-n\bar x^2}\nonumber\\
&=-2\lambda_2V,\label{eq:ODE_V}
\end{align}
where $\mathds{1}_n$ is a vector of size $n$ constituted of only ones and $x_\bot$ is the projection of $x$ on $\mathds{1}_n$. By the Comparison Lemma, $V$ is upper bounded by the solution of the right-hand side of \eqref{eq:ODE_V}, which yields \eqref{eq:consensus_continuous}.
The inequalities \eqref{eq:consensus_departure}-\eqref{eq:consensus_replacement} have been derived in \cite{monnoyer2020open}.
\end{proof}

\subsection{Sampling from graphons}
The space of all bounded symmetric measurable functions $W:[0,1]^2\rightarrow [0,1]$ is denoted by $\mathcal{W}$ and the elements of this space are called \textit{graphons}, whose name is a contraction of graph-function. 
The degree function of a graphon is \linebreak defined as:
$$
d(x):=\int_0^1W(x,y) dy,
$$
and its infimum is denoted by $\eta_W$.

A graphon $W$ can be used to generate random graphs using a sampling method \cite{lovasz2012large}. In this work, we will consider deterministic latent variables for the sampling, while the stochastic latent variables are left for future work \cite{avella2018centrality}.
\begin{definition}[Sampled Graph \cite{avella2018centrality}]\label{def:sampling}
Given a graphon $W$ and a size $n \in \mathbb{N}$, we say that the graph $G$ is sampled from $W$ if it is obtained through the following process:

\textit{1. Complete Weighted Graph $\bar G$}: 
let us fix deterministic latent variables $u_i=\frac{i}{n}$. We generate the complete weighted graph $\bar G$ with $n$ vertices, whose adjacency matrix is defined as: $\bar A(i,j)=W(u_i,u_j)$ for all $i,j\in\{ 1,\ldots,n\}$.

\textit{2. Simple Graph $G$}: from $\bar G$, we generate the simple graph $G$ with $n$ vertices by connecting each pair of distinct vertices $i\neq j$ with probability $\bar A(i,j)$ independently of the other edges.
\end{definition}

The Laplacian matrix of the complete weighted graph $\bar G$ is denoted by $\bar L=\bar D-\bar A$, where $\bar D=\text{diag}(\bar d_1,\ldots,\bar d_n)$ is the degree matrix. The eigenvalues of $\bar L$ are denoted as $0=\bar \lambda_1\le \bar\lambda_2\le \cdots\le \bar \lambda_n$ and the normalized versions are denoted as $\bar \mu_i=\bar\lambda_i/n$. The notation $L$ with the eigenvalues $0= \lambda_1\le \lambda_2\le \cdots\le  \lambda_n$ and the normalized eigenvalues $\bar \mu_i=\bar\lambda_i/n$ are used for the simple graphs $G$ obtained from a graphon.

Along this work, we consider that all the graph topologies associated with the dynamics \eqref{eq:linear_consensus} are sampled from a fixed graphon $W$ according to Definition~\ref{def:sampling}.

\section{Replacements}\label{sec:replacements}
 First, we consider the case of replacements of agents, which are used to model an arrival followed immediately by a departure \cite{hendrickx2016open,monnoyer2024random} or to approximate processes with similar rates of arrivals and departures such that the variations of the size of the system are almost negligible \cite{carletti2008birth,torok2013opinions,vizuete2024SIS}. 
 
 We consider that the time interval between two replacements is the same and it is given by $\Delta t^{k,k+1}=\gamma/n$, with $\gamma\ge 0$ for all $k\in\N$. This natural scaling with the size of the network $n$ is standard in the case of dense graphs sampled from graphons (see \cite{gao2019spectral,vizuete2020graphon,vizuete2022contributions} in the case of epidemics) which implies that the rate of replacements is proportional to the size of the network (i.e., more agents, more replacements).
Furthermore, we consider that during a replacement, the topology of the network changes and it is sampled again according to Definition~\ref{def:sampling}, independently of all the graph topologies sampled before the time of the replacement. An alternative procedure to generate the graph topology during a replacement would be to re-sample only the connections of the replaced agent. However, this creates technical difficulties since the graph topologies would be time-dependent.

\begin{theorem}\label{thm:replacement}
    For the linear consensus \eqref{eq:lim_n_k} in a system under replacements:
    \begin{equation}\label{eq:thm_replacement}
\limsup_{t\to\infty}\Ep{V(t)}\le\dfrac{\sigma^2(n^2-1)}{n\prt{n^2-(n^2-n-1)\Ep{e^{-2\gamma\mu_2}}}}.  
    \end{equation}
\end{theorem}

\vspace{2mm}

\begin{proof}
    From \eqref{eq:consensus_continuous}, we have that a time instant $t^{k+1}$ before the replacement, the disagreement satisfies:
    \begin{align}
    V(t^{k+1})&\le V^+(t^k)e^{-2\lambda_2(t^k)\gamma/n}\nonumber
    \\&=V^+(t^k)e^{-2\gamma\mu_2(t^k)},\label{eq:thm_replacements_1}
    \end{align}
    where $V^+(t^k)$ denotes the value of the disagreement at the time instant $k$ after the replacement and $\mu_2(t^k)$ is the second smallest normalized Laplacian eigenvalue of the network topology sampled at $t^k$. Since the graph sampled at $t^k$ is independent of all the graph topologies and replacements until $t^k$, we take the conditional expectation in \eqref{eq:thm_replacements_1} with respect to all the possible sampled graphs at $t^k$ and we obtain
    \begin{align}
        \Ep{V(t^{k+1})|V^+(t^k)}&\le V^+(t^k)\sum_{\ell}p_\ell e^{-2\gamma\mu_{2_\ell}(t^k)} \nonumber\\
         &=V^+(t^k)\Ep{e^{-2\gamma\mu_{2}(t^k)}}\label{eq:thm_replacements_temp}\\
        &=V^+(t^k)\Ep{e^{-2\gamma\mu_{2}}}, \label{eq:thm_replacements_2}
    \end{align}
    where $p_\ell$ denotes the probability associated with the graph $G_\ell$ sampled at the time instant $t^k$ (and also $\mu_{2_\ell}(t^k)$), and we remove the dependence on time of $\mu_2$ in \eqref{eq:thm_replacements_2} since the expected graph $\bar G$ and the sampling method are the same and independent for all the time instants corresponding to replacements such that $\Ep{e^{-2\gamma\mu_{2}(t^k)}}=\Ep{e^{-2\gamma\mu_{2}}}$ for all $k$. Then, we compute the total expectation in \eqref{eq:thm_replacements_2} and we get
    \begin{equation}\label{eq:thm_replacements_3}
        \Ep{V(t^{k+1})}\le \Ep{V^+(t^k)}\Ep{e^{-2\gamma\mu_{2}}}.
    \end{equation}
    Now, by using \eqref{eq:consensus_replacement}, we have that the disagreement at time $t^{k+1}$ after the replacement satisfies:
    \begin{equation}\label{eq:thm_replacements_4}
        \Ep{V^+(t^{k+1})}\le \alpha\Ep{V^+(t^k)}\Ep{e^{-2\gamma\mu_{2}}}+\beta,
    \end{equation}
where $\alpha=\frac{n^2-n-1}{n^2}$ and $\beta=\frac{n^2-1}{n^3}\sigma^2$. By the Comparison Lemma, we can guarantee that the disagreement is upper bounded by the solution of the right-hand side of \eqref{eq:thm_replacements_4}, which is an affine system of the form $x(k+1)=Ax(k)+u$ whose asymptotic behavior is given by  $\lim_{k\to\infty }x(k)=(I-A)^{-1}u$. Then, we have:
\begin{equation}\label{eq:thm_replacements_5}
\lim_{k\to\infty}\Ep{V^+(t^k)}\le \frac{\beta}{1-\alpha \Ep{e^{-2\gamma\mu_{2}}}}.    
\end{equation}
Finally, notice that \eqref{eq:thm_replacements_5} is valid for all $t$ since $V$ is always nonincreasing between two events $k$ and $k+1$ according to \eqref{eq:consensus_continuous},
and by replacing the values of $\alpha$ and $\beta$, we obtain the desired result.
\end{proof}
If the system is characterized only by replacements without modifying the graph topology of the network, the bound \eqref{eq:thm_replacement} just become:
$$
\limsup_{t\to\infty}\Ep{V(t)}\le\dfrac{\sigma^2(n^2-1)}{n\prt{n^2-(n^2-n-1)e^{-2\gamma\mu_2}}},  
$$
where $\mu_2$ is the second smallest normalized Laplacian eigenvalue of the fixed graph.

When there are only replacements (i.e., $\gamma\to0$) the bound \eqref{eq:thm_replacement} is obtained through the total expectation in \eqref{eq:consensus_replacement}, which yields:
\begin{equation}\label{eq:ODE_bound}
\Ep{V(t^{k+1})}\le \frac{n^2-n-1}{n^2}\Ep{V(t^k)}+
\frac{n^2-1}{n^3}\sigma^2,    
\end{equation}
since there is no continuous evolution of $V$. By using the Comparison Lemma and solving the right-hand side of \eqref{eq:ODE_bound}, the bound \eqref{eq:thm_replacement} will be given by $\frac{\sigma^2(n-1)}{n}$. 
In this scenario, when $n\to\infty$, the bound becomes $\sigma^2$ and when the number of agents is minimal (i.e., $n=1$), the bound becomes 0.

When there are no replacements (i.e., $\gamma\to\infty$), the bound \eqref{eq:thm_replacement} will be given by $\frac{\sigma^2(n^2-1)}{n^3}$. The value of the bound is not zero as it would be expected (just continuous evolution) since even if the interval tends to infinity, the bound remains based on the situation after a replacement where the consensus is perturbed by a new agent. 
That is why this bound coincides with the increase of $V$ during a replacement according to \eqref{eq:consensus_replacement} when  $V=0$. In this scenario, when $n\to\infty$, the bound becomes 0, because even if we have more agents, the impact of replacements in $V$ is almost negligible according to \eqref{eq:consensus_replacement}
and when the number of agents is minimal, the bound also becomes 0. The maximum value of the bound is obtained for $n=2$ and is given by $\frac{3}{8}\sigma^2$.

 Finally, let us analyze the behavior of the bound \eqref{eq:thm_replacement}, with respect to $\mu_2$. When the graphs are not so dense (i.e., $\mu_2\to 0$), the bound will be given by $\frac{\sigma^2(n-1)}{n}$ like in the case $\gamma\to0$, since the decrease of $V$ between two instants $t^k$ and $t^{k+1}$ will be negligible, and the bound will be mainly determined by the replacements. This is the maximum possible value of the bound. In the most dense graph (i.e., complete graph) we have $\mu_2=1$, and the bound will be given by $\frac{\sigma^2(n^2-1)}{n\prt{n^2-(n^2-n-1)e^{-2\gamma}}}$, which is the minimum possible value of the bound.

\section{Arrivals and departures}

\subsection{Process for arrivals and departures}

We consider that at each $k$-th event, there is a departure with probability $p_D(t^k)=\tau (n(t^k)-n_{m})$ or an arrival with probability $p_A(t^k)=1-p_D( t^k)=\tau(n_M-n(t^k))$, with $\tau$ satisfying $\tau (n_{M}-n_{m})=1$, where $n_{m}\in\N$ is the minimum number of agents, $n_M\in\N$ is the maximum number of agents, and $n_M>n_m$. Notice that at each $k$, there is either an arrival or a departure, and the two events cannot happen at the same time so that they are mutually exclusive. The process is well defined since when the number of agents is $n_{M}$, the probability of departures is 1, while when the number of agents is $n_{m}$, the probability of arrivals is 1. Therefore, the current number of agents satisfies $n(t)\in [n_{m},n_{M}]$.

At this preliminary stage, we consider this model with a bounded number of agents and probabilities that are linear on the current number of agents since it is technically simpler and allows us to represent variations in the number of agents while remaining bounded. Nevertheless, for future work, we could use more complex models including independent Poisson processes like in \cite{monnoyer2020open}. Notice that this process used for arrivals and departures resembles a birth-death process with bounded size and it is depicted in Fig.~\ref{fig:bd_process}. A similar process has been used in \cite{monnoyer2020open}, with two different Poisson processes with unbounded size. In our case, there is only one process and the size is limited \footnote{When the size of the system is bounded, a potential approach for the analysis is to use \emph{pseudo open multi-agent systems} based on a finite superset with all the possible agents in the system, where the change of dimension is modeled by activations/deactivations of agents \cite{varma2018open,vizuete2020influence,hsieh2021optimization}. However, in our case, even if the size of the system is bounded, during an arrival, we consider that the arriving agent is completely new, and when there is a departure, we assume that the leaving agent will not return to the system, so that the pool of all possible agents is unbounded.}.

\begin{figure}[!ht]
\centering
{ \includegraphics[width=\linewidth]{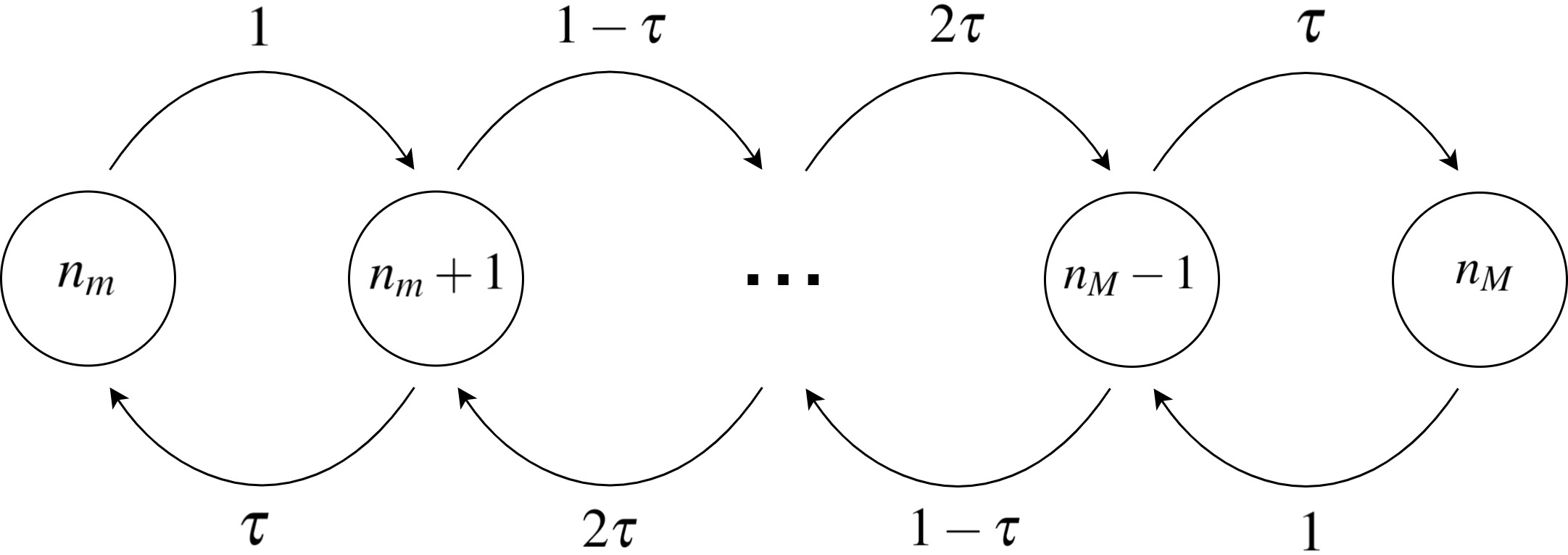}}
\caption{Stochastic process used to model the evolution of the number of agents $n(t^k)$. The process resembles a birth-death process with a minimum number of agents $n_m$ and a maximum number of agents $n_M$ such that $n(t^k)\in[n_m,n_M]$ for all $k\in\N$.}\label{fig:bd_process}
\end{figure}

\begin{proposition}\label{prop:expected_n}
    For the process of arrivals and departures:
    \begin{equation}\label{eq:lim_n_k}
\lim_{t\to\infty}\Ep{n(t)}=\frac{\nM+\nm}{2}.     
    \end{equation}
\end{proposition}

\vspace{2mm}

\begin{proof}
    The conditional expectation of the number of agents at an event $t^{k+1}$  satisfies:
    \begin{align*}
    \Ep{n(t^{k+1})|n(k)}&=\tau(n_M-n(t^k))(n(t^k)+1)+\\&\quad\;\tau( n(t^k)- n_{m})(n(t^k)-1)\\
    &=(1-2\tau)n( t^k)+\tau(n_M+n_m).
    \end{align*}
By taking the total expectation, we obtain:
\begin{equation}\label{eq:dynamics_n_prop}
\Ep{n(t^{k+1})}=(1-2\tau)\Ep{n(t^k)}+\tau(n_M+n_m),    
\end{equation}
which is an affine system of the form $x(k+1)=Ax(k)+u$, whose asymptotic behavior is given by $\lim_{k\to\infty }x(k)=(I-A)^{-1}u$, and the proof is completed.
\end{proof}

\subsection{Disagreement subject to arrivals and departures}

Following the same approach as in the case of replacements in Section~\ref{sec:replacements}, we consider that the time interval between two events at $k$ and $k+1$ is given by $\Delta t^{k,k+1}=\gamma/n(t^k)$, with $\gamma\ge 0$ for all $k\in\N$. Furthermore, we consider that during an arrival or a departure, the topology of the network is sampled again according to Definition~\ref{def:sampling}, independently of all the graph topologies sampled before the time instant of the arrival/departure.

\begin{theorem}\label{thm:arrivals}
    For the linear consensus \eqref{eq:linear_consensus} in a system under arrivals and departures with $n_M>3$:
    \begin{multline}\label{eq:thm_arrivals}
      \limsup_{t\to\infty}\Ep{V(t)}\le \\
      \frac{\sigma^2(\nM-1)^2}{2(\nm+1)\prt{(\nM-1)^2-\nM(\nM-2)\Ep{e^{-2\gamma\mu_2}}_{M}}},  
    \end{multline}
    where $\Ep{e^{-2\gamma\mu_{2}}}_M:=\max_{n\in[n_m,n_M]} \Ep{e^{-2\gamma\mu_{2}^{(n)}}}$ and each $\Ep{e^{-2\gamma\mu_{2}^{(n)}}}$ corresponds to the expectation among all the possible $\mu_2$ given a number of agents $n$.
\end{theorem}
\begin{proof}
    From \eqref{eq:consensus_continuous}, we have that a time instant $t^{k+1}$ before the arrival/departure, the disagreement satisfies:
    \begin{equation}\label{eq:thm_arrivals_1}
    V(t^{k+1})\le V^+(t^k)e^{-2\gamma\mu_2(t^k)},
    \end{equation}
    and by following a reasoning similar to the proof of Theorem~\ref{thm:replacement} until \eqref{eq:thm_replacements_temp}, we compute the conditional expectation obtaining:  \begin{equation}\label{eq:thm_arrivals_2}      \Ep{V(t^{k+1})|V^+(t^k)}\le 
        V^+(t^k)\Ep{e^{-2\gamma\mu_{2}(t^k)}}. 
    \end{equation}
   Since $n(t^k)$ satisfies $n(t^k)\in[n_m,n_M]$, and each potential number of agents $n_\ell$ has an associated $\Ep{e^{-2\gamma\mu_{2}^{(n_\ell)}}}$, we consider the lowest decay rate $\Ep{e^{-2\gamma\mu_{2}}}_M:=\max_{n\in[n_m,n_M]} \Ep{e^{-2\gamma\mu_{2}^{(n)}}}$ so that \eqref{eq:thm_arrivals_2} can be upper bounded by:
   \begin{equation*}
       \Ep{V(t^{k+1})|V^+(t^k)}\le 
        V^+(t^k)\Ep{e^{-2\gamma\mu_{2}}}_M,
   \end{equation*}
   whose total expectation yields:
   \begin{equation}\label{eq:thm_arrivals_3}
       \Ep{V(t^{k+1})}\le 
        \Ep{V^+(t^k)}\Ep{e^{-2\gamma\mu_{2}}}_M.
   \end{equation}
By using \eqref{eq:consensus_departure} and \eqref{eq:consensus_arrival} with the lower and upper bounds $n_m$ and $n_M$ respectively, we compute the conditional expectation of the disagreement at the time instant $t^{k+1}$ after the arrival/departure, given $V(t^{k+1})$ and $n(t^k)$:

\vspace{-4mm}

\small
\begin{align}
     &\Ep{V^+(t^{k+1})|V(t^{k+1}),n(t^k)}\nonumber\\ &\le p_D\prt{1-\frac{1}{(n_{M}-1)^2}}V(t^{k+1})+p_A\prt{\frac{\nM}{\nM+1}V(t^{k+1})+\frac{\sigma^2}{\nm+1}}\nonumber\\
     &=p_D\alpha V(t^{k+1})+p_A(\beta V(t^{k+1})+\xi)\nonumber\\
     &=\tau(n(t^k)-\nm)\alpha V(t^{k+1})+\tau(n_M-n(t^k))(\beta V(t^{k+1})+\xi)\nonumber\\
     &=(\alpha-\beta)\tau V(t^{k+1})n(t^k)+\tau(\beta n_M-\alpha n_m)V(t^{k+1})-\tau\xi n(t^k)+\nonumber\\
     &\quad\; \xi\tau n_M,\label{eq:thm_arrivals_4}
\end{align}

\normalsize

\noindent where $\alpha=1-\frac{1}{(n_{M}-1)^2}$, $\beta=\frac{\nM}{\nM+1}$ and $\xi=\frac{\sigma^2}{\nm+1}$. Notice that $\alpha>\beta$ if $\frac{n_M(n_M-3)}{(n_M+1)(n_M-1)^2}>0$, which holds for $n_M>3$. Then, by using the bound $(\alpha-\beta)\tau V(t^{k+1})n(t^k)\le (\alpha-\beta)\tau V(t^{k+1})n_M$ and computing the total expectation in \eqref{eq:thm_arrivals_4} we obtain:

\small

\vspace{-4mm}

\begin{align}
&\Ep{V^+(t^{k+1})}\nonumber
\\&\le \tau\nM(\alpha-\beta) \Ep{V(t^{k+1})}+\tau(\beta n_M-\alpha n_m)\Ep{V(t^{k+1})}-\nonumber\\ &\quad\;\tau\xi \Ep{n(t^k)}+\xi\tau n_M\nonumber\\
&=\alpha \Ep{V(t^{k+1}))}-\phi \Ep{n(t^k)}+\zeta,\label{eq:dynamics_V00}
\end{align}

\normalsize

\noindent where $\phi=\tau\xi$ and $\zeta=\xi\tau n_M$. By using \eqref{eq:thm_arrivals_3} in \eqref{eq:dynamics_V00} we get:
\begin{equation}\label{eq:dynamics_V0}
    \Ep{V^+(t^{k+1})}\le \alpha\Ep{V^+(t^k)}\Ep{e^{-2\gamma\mu_{2}}}_M-\phi \Ep{n(t^k)}+\zeta.
\end{equation}
From \eqref{eq:lim_n_k}, we have that for all $\epsilon>0$, there exists a $k^*\in\N$ such that $\abs{n(t^k)-\frac{n_M+n_m}{2}}<\epsilon$ for all $k>k^*$. This implies that \eqref{eq:dynamics_V0} satisfies:

\small

\vspace{-4mm}

\begin{equation*}
    \Ep{V^+(t^{k+1})}\le \alpha\Ep{V^+(t^k)}\Ep{e^{-2\gamma\mu_{2}}}_M-\phi \prt{\frac{n_M+n_m}{2}-\epsilon}+\zeta,
\end{equation*}

\normalsize

\noindent and by taking the limit $k\to\infty$ we get:
\begin{align}
    \lim_{k\to\infty} \Ep{V^+(t^{k+1})}&\le \frac{\zeta-\phi \prt{\frac{n_M+n_m}{2}-\epsilon}}{1-\alpha\Ep{V^+(t^k)}\Ep{e^{-2\gamma\mu_{2}}}_M}\nonumber\\
    &=\frac{\xi\tau n_M-\tau\xi\frac{n_M+n_m}{2}}{1-\alpha\Ep{V^+(t^k)}\Ep{e^{-2\gamma\mu_{2}}}_M}+\hat \epsilon\nonumber\\
    &=\frac{\xi}{2}\frac{1}{1-\alpha\Ep{V^+(t^k)}\Ep{e^{-2\gamma\mu_{2}}}_M}+\hat\epsilon\label{eq:V_epsilon},
\end{align}
which is valid for all $\hat\epsilon>0$. 
   According to \eqref{eq:consensus_continuous}, the disagreement is always nonincreasing between two time instants $t^k$ and $t^{k+1}$, so that \eqref{eq:V_epsilon} is valid for all time $t$ and, by replacing the value of $\xi$ and $\alpha$ we obtain the desired result.
\end{proof}

When there are only arrivals/departures (i.e., $\gamma\to0$), the bound \eqref{eq:thm_arrivals} is given by $\frac{\sigma^2(\nM-1)^2}{2(\nm+1)}$. If $\nm=\nM=n$ the bound becomes $\frac{\sigma^2(n-1)^2}{n+1}$. If $n=1$ the bound becomes 0 and when $n\to\infty$ the bound diverges. 

When there are no arrivals/departures (i.e., $\gamma \to\infty$) , the bound \eqref{eq:thm_arrivals} is given by $\frac{\sigma^2}{2(\nm+1)}$. Similarly to the replacement case, the value of the bound is not zero as it would be expected since the bound remains based on the situation after an arrival/departure where the consensus is perturbed by a new agent even if the interval tends to infinity. This value coincides with the expected increase of $V$ during an arrival/departure when $V=0$ and $n=\frac{n_M+n_m}{2}$, and corresponds to the right-hand side of \eqref{eq:dynamics_V0} with $\Ep{V^+(t^k)}=0$ and $\Ep{n(t^k)}=\frac{n_M+n_m}{2}$.
When $\nm=1$, the bound becomes $\frac{\sigma^2}{4}$ and when $n\to\infty$, the bound becomes 0.

When $n_M\to\infty$, the bound \eqref{eq:thm_arrivals} is given by $\frac{\sigma^2}{2(n_m+1)(1-\Ep{e^{-2\gamma\mu_{2}}}_M)}$. Since for graphs sampled from graphons, $\mu_2$ converges to a value different from zero \cite{vizuete2021laplacian}, for large values of $n_M$, the bound \eqref{eq:thm_arrivals} is basically determined by $n_m$.

Finally, regarding the dependence of \eqref{eq:thm_arrivals} on $\mu_2$, the maximum value of the bound corresponds to $\mu_2\to 0$ (i.e., not so dense graphs) and is given by $\frac{\sigma^2(n_M-1)^2}{2(n_m+1)}$, which coincides with the case $\gamma\to 0$. The minimum value of the bound corresponds to $\mu_2=1$ (i.e., complete graph) and is given by $\frac{\sigma^2(\nM-1)^2}{2(\nm+1)\prt{(\nM-1)^2-\nM(\nM-2)e^{-2\gamma}}}$.

\section{Impact of the graphon structure}

\subsection{Estimation of $\Ep{e^{-2\gamma\mu_{2}}}$}

The upper bound for replacements derived in Theorem~\ref{thm:replacement} is a function of $\Ep{e^{-2\gamma\mu_{2}}}$ given a fixed number of agents $n$, which by definition should imply the computation of the spectrum of $2^{n(n-1)/2}$  matrices \footnote{Number of ways to interconnect $n$ nodes to form a graph \cite{harary2014graphical}} of size $n$. In the case of the upper bound for arrivals/departures in Theorem~\ref{thm:arrivals}, the computation of $\Ep{e^{-2\gamma\mu_{2}}}_M$ by definition requires the computation of the spectrum of a number of matrices in the order of $(n_M-n_m)2^{n_M(n_M-1)/2}$. In both scenarios, this could be hard to compute, especially when the number of agents is large (which is common in open multi-agent systems).
For this reason, we are interested in an upper bound as a function of $\bar\mu_2$, which we recall is the second smallest normalized Laplacian eigenvalue of $\bar G$. This bound will be obtained for a class of graphons that is wide enough to be relevant for many applications and that has been introduced in \cite{avella2018centrality}.

\begin{definition}[Piecewise Lipschitz graphon]
Graphon $W$ is said to be \textit{piecewise Lipschitz} if there exists a constant $L$ and a sequence of non-overlapping intervals $I_k=[\alpha_{k-1},\alpha_k)$ defined by $0=\alpha_0 < \cdots <\alpha_{K+1}=1$, for an  integer $K\ge 0$ such that for any $k, \ell$, any set $I_{k\ell}=I_k \times I_\ell$ and pairs $(x_1,y_1) \textrm{ and } (x_2,y_2)\in I_{k\ell}$ we have that:
$$
\vert W(x_1,y_1)-W(x_2,y_2)\vert \leq L(\vert x_1-x_2 \vert+\vert y_1-y_2 \vert) . 
$$
\end{definition}

\vspace{2mm}

\begin{definition}[Large enough $n$]\label{largeN}
Given a piecewise Lipschitz graphon $W$ and $\epsilon<e^{-1}$, we will say that $n$ is \textit{large enough} if $n$ satisfies the following conditions:
\begin{subequations}
\begin{equation}
\dfrac{2}{n}<\min_{k\in \{1,\ldots , K+1\}}(\alpha_k-\alpha_{k-1}) ; \label{eq_condi1}
\end{equation}
\begin{equation}
\dfrac{1}{n}\log\left(\dfrac{2n}{\epsilon}\right)+\dfrac{1}{n}(2K+3L)<\max_{x}d(x) ; \label{eq_condi2}
\end{equation}
\begin{equation}
ne^{-n/5}<\epsilon; \label{eq_condi3}
\end{equation}
\begin{equation}
9\log(2en)<n. \label{eq_condi4}
\end{equation}
\end{subequations}
\end{definition}

\vspace{2mm}

\begin{theorem}\label{thm:exponential_lambda}
    Given a piecewise Lipschitz graphon $W$, for $n$ large enough, we have:
    \begin{equation}
        \Ep{e^{-2\gamma\mu_2}}\le e^{-2\gamma \bar \mu_2}\prt{e^{6\gamma\sqrt{\frac{\log(2en)}{n}}}+\Psi(n)},\label{eq:thm_arrivals_eq1}    
    \end{equation}

\vspace{-5mm}

\small

\begin{multline*}
\!\!\!\!\!\!\!\!\!\!  \text{where }  \Psi(n):=\\12\gamma\sqrt{\pi n}e^{9\gamma^2/n}\prt{\sqrt{1-e^{-\frac{4}{9\pi n}(n-9\gamma)^2}}-\sqrt{1-e^{-\frac{1}{n}(\sqrt{n\log(2en)}-3\gamma)^2}}},  
\end{multline*}

\normalsize

\vspace{-1mm}

\noindent and $\Psi(n)=O(\frac{1}{\sqrt{n}})$.

\end{theorem}
Before presenting the proof of Theorem~\ref{thm:exponential_lambda}, we recall a result that will be used in the proof.
\begin{lemma}[\cite{vizuete2021laplacian,garin2024corrections}]\label{lemma:second_eig}
    Given a graphon $W$, if $\bar d_{(n)}>\frac{4}{9}\log(2n/\epsilon)$, with probability at least $1-2\epsilon$ the normalized eigenvalues $\mu_i$ and $\bar \mu_i$ of the Laplacian matrices $L$ and $\bar L$  respectively, satisfy:
    $$
    \max_{i=1,\ldots,n}\abs{\mu_i-\bar \mu_i}\le 3\sqrt{\frac{\log(2n/\epsilon)}{n}},
    $$
    where $\bar d_{(n)}$ is the maximum expected degree.
\end{lemma}
\begin{myproof}{Theorem~\ref{thm:exponential_lambda}}
    The condition $\bar d_{(n)}>\frac{4}{9}\log(2n/\epsilon)$ of Lemma~\ref{lemma:second_eig} implies that $\epsilon$ must satisfy
    $
    \epsilon>2ne^{-\frac{9\bar d_{(n)}}{4}}.
    $
    From \cite{avella2018centrality}, we have that for $n$ large enough, the condition $
    \epsilon>2ne^{-\frac{9\bar d_{(n)}}{4}}
    $ is always satisfied when $\epsilon\in(ne^{-n/5},e^{-1})$, which does not depend on $\bar d_{(n)}$.
    Then, by applying Lemma~\ref{lemma:second_eig} to the second smallest eigenvalue $\mu_2$ we have that for $\epsilon\in(ne^{-n/5},e^{-1})$:
    $$  \Pr{\brk{2\gamma\abs{\mu_2-\bar \mu_2}>6\gamma\sqrt{\frac{\log(2n/\epsilon)}{n}}}}\le 2\epsilon, 
    $$
    which implies:
\begin{equation}\label{eq:prob_exp}
\Pr{\brk{e^{2\gamma\abs{\mu_2-\bar \mu_2}}>e^{6\gamma\sqrt{\frac{\log(2n/\epsilon)}{n}}}}}\le 2\epsilon.    
    \end{equation}
    By making the change of variable $z=e^{6\gamma\sqrt{\frac{\log(2n/\epsilon)}{n}}}$ we have:
    $
    \epsilon=2ne^{-\frac{n\log^2 z}{36\gamma^2}},
    $
    such that \eqref{eq:prob_exp} can be expressed as: \begin{equation}\label{eq:prob_exp_2}
        \Pr{\brk{e^{2\gamma\abs{\mu_2-\bar \mu_2}}>z}}\le 4ne^{-\frac{n\log^2 z}{36\gamma^2}},
    \end{equation}
    for $z\in\prt{e^{6\gamma\sqrt{\frac{\log(2en)}{n}}},e^{6\gamma\sqrt{\frac{n+5\log 2}{5n}}}}$, where this interval of validity of \eqref{eq:prob_exp_2} is obtained from the previous interval of validiy $\epsilon\in(ne^{-n/5},e^{-1})$ applied to the  variable $z$.
    
\noindent    Now, we focus on $\Ep{e^{-2\gamma\mu_2+2\gamma\bar\mu_2}}$, which satisfies:
    \begin{align}
        \Ep{e^{-2\gamma\mu_2+2\gamma\bar\mu_2}}&\le\Ep{e^{2\gamma\abs{-\mu_2+\bar\mu_2}}}\nonumber\\
        &=\int_0^\infty \Pr{\brk{e^{2\gamma\abs{\mu_2-\bar \mu_2}}>z}}dz\nonumber\\
        &=\int_0^{e^{2\gamma}} \Pr{\brk{e^{2\gamma\abs{\mu_2-\bar \mu_2}}>z}}dz\label{eq:integral_1},
    \end{align}
    where we use the fact that $\abs{\mu_2-\bar\mu_2}\in [0,1]$ (since they are normalized eigenvalues \footnote{The largest Laplacian eigenvalue $\lambda_n$ of a simple graph satisfies $\lambda_n\le n $ \cite{van2023graph}. By using Jensen's inequality we have $\norm{\bar L}\le \Ep{\norm{L}}\le n$, where $\norm{\cdot}$ is the spectral norm, which implies $\bar \lambda_n\le n$.}, such that $\Pr{\brk{e^{2\gamma\abs{\mu_2-\bar \mu_2}}>z}}=0$ for $z>e^{2\gamma}$. Now, notice that $e^{6\gamma\sqrt{\frac{n+5\log 2}{5n}}}>e^{2\gamma}$ for all $n$, and $e^{6\gamma\sqrt{\frac{\log(2en)}{n}}}<e^{2\gamma} $ for $n$ large enough (see \eqref{eq_condi4}), such that \eqref{eq:integral_1} can be expressed as:
    \begin{align*}
        \Ep{e^{-2\gamma\mu_2+2\gamma\bar\mu_2}}&\le\int_0^{e^{6\gamma\sqrt{\frac{\log(2en)}{n}}}} \Pr{\brk{e^{2\gamma\abs{\mu_2-\bar \mu_2}}>z}}dz+\\
        &\quad\;\int_{e^{6\gamma\sqrt{\frac{\log(2en)}{n}}}}^{e^{2\gamma}} \Pr{\brk{e^{2\gamma\abs{\mu_2-\bar \mu_2}}>z}}dz\\
        &\le \int_0^{e^{6\gamma\sqrt{\frac{\log(2en)}{n}}}} 1dz+\int_{e^{6\gamma\sqrt{\frac{\log(2en)}{n}}}}^{e^{2\gamma}} 4ne^{-\frac{n\log^2 z}{36\gamma^2}}dz\\
        &=e^{6\gamma\sqrt{\frac{\log(2en)}{n}}}+4n\int_{e^{6\gamma\sqrt{\frac{\log(2en)}{n}}}}^{e^{2\gamma}} e^{-\frac{n\log^2 z}{36\gamma^2}}dz,
    \end{align*}
    where we used \eqref{eq:prob_exp_2}. Now, we focus on the integral $\int_{e^{6\gamma\sqrt{\frac{\log(2en)}{n}}}}^{e^{2\gamma}} e^{-\frac{n\log^2 z}{36\gamma^2}}dz$. We make the change of variable $\alpha=\log z$, which gives $dz=e^\alpha d\alpha$ and we get:
    $$
\int_{e^{6\gamma\sqrt{\frac{\log(2en)}{n}}}}^{e^{2\gamma}} e^{-\frac{n\log^2 z}{36\gamma^2}}dz=
\int_{6\gamma\sqrt{\frac{\log(2en)}{n}}}^{2\gamma} e^{-\frac{n}{36\gamma^2}\alpha^2+\alpha}d\alpha.
    $$
    Then, we use $\int e^{-ax^2+bx}dx=\frac{\sqrt{\pi}e^{b^2/(4a)}\erf\prt{\frac{2ax-b}{2\sqrt{a}}}}{2\sqrt{a}}+C $ for $a>0$, where $\erf(x)=\frac{2}{\sqrt{\pi}}\int_0^xe^{-t^2}dt$ is the error function, to obtain:
    \begin{multline*}
       4n\int_{e^{6\gamma\sqrt{\frac{\log(2en)}{n}}}}^{e^{2\gamma}} e^{-\frac{n\log^2 z}{36\gamma^2}}dz=\\ 
       12\gamma\sqrt{\pi n}e^{9\gamma^2/n}\prt{\erf\prt{\frac{n-9\gamma}{3\sqrt{n}}}-\erf\prt{\frac{\sqrt{n\log(2en)}-3\gamma}{\sqrt{n}}}}.
    \end{multline*}
    By using the fact that the error function satisfies the following bounds $\sqrt{1-e^{-x^2}}\le\erf(x)\le\sqrt{1-e^{-\frac{4}{\pi}x^2}}$ \cite{chu1955bounds}, we obtain \eqref{eq:thm_arrivals_eq1}.
    
    \noindent Now, we analyze the asymptotic behavior of the bound \eqref{eq:thm_arrivals_eq1} for $n\to\infty$. By using the Binomial approximation \footnote{$(1+x)^\alpha\approx 1+\alpha x $ for $\abs{x}<1$ and $\abs{\alpha x}\ll 1$.}, we have:

    \small

    \vspace{-4mm}
    
    \begin{align*}
    \Psi(n)&\approx 12\gamma\sqrt{\pi n}e^{9\gamma^2/n}\prt{\frac{1}{2}e^{-\frac{1}{n}(\sqrt{n\log(2en)}-3\gamma)^2}-\frac{1}{2}e^{-\frac{4}{9\pi n}(n-9\gamma)^2}}\\
    &\approx 6\gamma\sqrt{\pi}e^{9\gamma^2/n}\prt{\sqrt{n}e^{-\frac{1}{n}(\sqrt{n\log(2en)}-3\gamma)^2}-\sqrt{n}e^{-\frac{4}{9\pi n}(n-9\gamma)^2}},
    \end{align*}

    \normalsize

    \vspace{-2mm}
    
    \noindent where $\sqrt{n}e^{-\frac{1}{n}(\sqrt{n\log(2en)}-3\gamma)^2}$ goes to zero as $O(\frac{1}{\sqrt{n}})$ and $\sqrt{n}e^{-\frac{4}{9\pi n}(n-9\gamma)^2}$ decays to zero as $O(\sqrt{n}e^{-n})$. This implies that the rate of convergence to zero of $\Psi(n)$ is $O(\frac{1}{\sqrt{n}})$. Finally, notice that $e^{6\gamma\sqrt{\frac{\log(2en)}{n}}}$ converges to 1 when $n\to\infty$.
\end{myproof}

The bounds of Theorems~\ref{thm:replacement} and \ref{thm:arrivals} are valid if $e^{-2\gamma \bar \mu_2}\prt{e^{6\gamma\sqrt{\frac{\log(2en)}{n}}}+\Psi(n)}<1$, which is expected to hold for a large number of agents.

\subsection{Computation of $\bar \mu_2$ for stochastic block model (SBM) graphons }

With Theorem~\ref{thm:exponential_lambda}, the computation of the bound of Theorem~\ref{thm:replacement} is reduced to the computation of the spectrum of just one matrix of size $n$, while the computation of the bound of Theorem~\ref{thm:arrivals} requires the computation of the spectrum of $n_M-n_m$ matrices of order $n_M$. This can still be computationally inefficient for a large number of agents and that is why, in this section, we will focus on the computation of the spectrum of stochastic block model (SBM) graphons that are used to model community structures \cite{karrer2011stochastic} and also experimentally for the estimation of graphons \cite{airoldi2013stochastic}.

A SBM graphon is a piecewise constant graphon, and thus piecewise Lipschitz, defined as \cite{avella2018centrality}:
\begin{equation}\label{eq:SBM_definition}
  W_{\SBM}(x,y):=\sum_{i=1}^m\sum_{j=1}^m
P_{ij}\chi_{B_i}(x)\chi_{B_j}(y),  
\end{equation}
where $\chi_B(x)$ is the indicator function, $P_{ij}\in[0,1]$, $P_{ij}=P_{ji}$, $\bigcup_{i=1}^mB_i=[0,1]$ and $B_i\cap B_j=\emptyset$. Similarly to \cite{avella2018centrality}, let us define the probability matrix  $P_\SBM:=[P_{ij}]\in\R^{m\times m}$ and the matrix $E_\SBM:=\text{diag}(n_{B_1},\ldots,n_{B_m})$, where $n_{B_i}$ is the number of latent variables in the interval $B_i$. We define the adjacency-type matrix of the SBM graphon as:
$$
A_{\SBM}:=\frac{1}{n}P_\SBM E_\SBM.$$
Notice that the matrix $A_{\SBM}$ is not necessarily symmetric even if the graphon is symmetric.
Related to $A_{\SBM}$, we define the degree matrix $D_\SBM:=\text{diag}(\delta_1,\ldots,\delta_m)$ where each $\delta_i$ is the $i$-th row sum of $A_\SBM$. We denote by $\delta_{\min}$ the minimum degree of $D_{\SBM}$. Then, we define the Laplacian-type matrix of the SBM graphon as:
$$
L_\SBM:=D_{\SBM}-A_{\SBM}.
$$

\begin{proposition}[SBM graphon]\label{prop:SBM}
    For a SBM graphon $W_\SBM$ with $n$ deterministic latent variables:
    $$
\bar\mu_2=\min(\lambda_2( L_{\SBM}),\delta_{\min}),
    $$
    where $\lambda_2(L_{\SBM})$ is the second smallest eigenvalue of $L_{\SBM}$.
\end{proposition}

\vspace{2mm}

\begin{proof}
    Let us consider a particular decomposition of the Laplacian matrix of the complete weighted graph $\bar G$:
    $$
    \bar L=\hat L+L_D,
    $$
    where $\hat L$ is a matrix of the form
$$
\hat L=
\begin{bmatrix}
D_{B_1,B_1} & -\bar A_{B_1B_2} & \cdots & -\bar A_{B_1B_m}\\
-\bar A_{B_2B_1} & D_{B_2,B_2} & \cdots & -\bar A_{B_2B_m}\\
\vdots  & \vdots & \ddots & \vdots\\
-\bar A_{B_mB_1} & -\bar A_{B_mB_2}  & \cdots & D_{B_m,B_m}
\end{bmatrix},
$$
where each block $\bar A_{B_iB_j}=P_{ij}\ones_{n_{B_i}}\ones_{n_{B_j}}^T\in\R^{n_{B_i}\times n_{B_j}}$ is determined by the adjacency matrix $\bar A$ and each block $D_{B_iB_i}=\frac{1}{n_{B_i}}\sum_{j\neq i}P_{ij}n_{B_j}\ones_{n_{B_i}}\ones_{n_{B_i}}^T\in\R^{n_{B_i}\times n_{B_i}}$ is a block with the same entries such that the sum of each row is 0. 
The matrix $L_D$ is defined as $L_D=\bar L-\hat L$. Notice that $L_D$ is a block diagonal matrix where each block correspond to the Laplacian matrix of a complete weighted graph. A block $L_{D_{B_i}}$ has off-diagonal weights $-P_{ii}-\frac{1}{n_{B_i}}\sum_{j\neq i}P_{ij}n_{B_j}$, and the weights of the diagonal are $\frac{n_{B_i}-1}{n_{B_i}}\sum_{j\neq i}P_{ij}n_{B_j}+P_{ii}(n_{B_i}-1)$.

First, let us analyze the eigenvalues of the matrix $L_D$. Each block $L_{D_{B_i}}$ has $n_{B_i}-1$ eigenvalues with magnitude $\sum_{j\neq i}P_{ij}n_{B_j}$ associated to eigenvectors with a structure such that only the entries corresponding to the block are non zero and the sum of all the entries are zero (since they are eigenvalues and eigenvectors of a complete graph). Let us consider an arbitrary nonzero eigenvalue $\lambda_{L_D}$ with an eigenvector $v_{L_D}$. Notice that $\hat Lv_{L_D}=0$, since the block of $\hat L$ that multiplies the nonzero entries of $v_{L_D}$ has the same entries. This implies that all the eigenvectors $v_{L_D}$ of $L_D$ associated with nonzero eigenvalues $\lambda_{L_D}$ satisfy $v_{L_D}\in \text{ker}(\hat L)$. Therefore, the matrix $\bar L$ has for each $i\in\{1,\ldots,m\}$, $n_{B_i}-1$ eigenvalues of magnitude $\sum_{j\neq i}P_{ij}n_{B_j}$, which correspond to the degrees (eigenvalues)  of the matrix $D_{\SBM}$ multiplied \linebreak by $n$.

\noindent Now, let us analyze the eigenvalues of the matrix $\hat L$. An eigenvector $v_{\hat L}$ associated with a nonzero eigenvalue $\lambda_{\hat L}$ of $\hat L$ must have a structure of the form $v_{\hat L}=[v_{\hat L}^{(1)}\ones_{n_{B_1}}^T,v_{\hat L}^{(2)}\ones_{n_{B_2}}^T,\dots,v_{\hat L}^{(m)}\ones_{n_{B_m}}^T]^T$
where each $v_{\hat L}^{(i)}$ is repeated $n_{B_i}$ times. In addition, notice that $L_Dv_{\hat L}=0$, since each block $L_{D_{B_i}}$ multiplies the same entry $v_{\hat L}$. This implies that all the eigenvectors $v_{\hat L}$ of $\hat L$ associated with nonzero eigenvalues $\lambda_{\hat L}$ satisfy $v_{\hat L}\in \text{ker}(L_D)$ and hence, the nonzero eigenvalues of $\hat L$ are also eigenvalues of $\bar L$. This implies that the eigenvalue $\lambda_{\hat L}$ associated with the eigenvector $v_{\hat L}$ must be also an eigenvalue of the smaller matrix:
$$
\begin{bmatrix}
\sum_{j\neq 1}P_{1j}n_{B_j} & -P_{12} n_{B_2}  & \cdots & - P_{1m}n_{B_m}\\
-P_{12}n_{B_1} & \sum_{j\neq 2}P_{2j}n_{B_j} &  \cdots & -P_{2m}n_{B_m}\\
\vdots  & \vdots & \ddots & \vdots\\
-P_{1m}n_{B_1} & -P_{2m}n_{B_2} & \cdots & \sum_{j\neq m}P_{mj}n_{B_m}
\end{bmatrix},    
$$
which corresponds to the matrix $L_{\SBM}$ multiplied by $n$. Then, we have that $m-1$ normalized eigenvalues of $\bar L$ are determined by the eigenvalues of $L_{\SBM}$, and $n-m$ normalized eigenvalues of $\bar L$ are determined by the eigenvalues of $D_{\SBM}$. The last eigenvalue of $\bar L$ is the trivial zero of any Laplacian matrix.
\end{proof}

The appearance of the minimum degree $\delta_{\min}$ as a a possible value for $\bar \mu_2$ in Proposition~\ref{prop:SBM} has a clear link with the convergence of normalized Laplacian eigenvalues to the spectrum of infinite dimensional operators associated with graphons. For a graphon $W$, we can define the Laplacian operator $\mathcal{L}_W:L^2[0,1]\to L^2[0,1]$ as:
\begin{equation}\label{eq:laplacian-operator}
(\mathcal{L}_Wf)(x):=d(x)f(x)- \int_0^1 W(x,y)f(y)dy.   
\end{equation}
This operator has an essential spectrum located in the range of the degree function $d(x)$ and a finite number of isolated eigenvalues $\kappa_i$ (see \cite{vizuete2021laplacian} for more details).

\begin{proposition}[\cite{vizuete2021laplacian}]\label{convergence}
Let $W$ be a continuous graphon.
Let $M$  be the number of isolated eigenvalues of $\mathcal{L}_W$ in the interval $[0,\eta_W]$,  
counted with their multiplicities, and let $0 = \kappa_1 \le \kappa_2 \le \dots \le \kappa_M$ be such eigenvalues.
Define $\rho = \kappa_2$ if $M \ge 2$, and $\rho = \eta_W$ if $M = 1$.
Then, 
$$
\lim_{n \to \infty} \mu_2 = \rho, \qquad \mathrm{a.s.}
$$
\end{proposition}

\vspace{2mm}

Proposition~\ref{convergence} provides the limit of the sequence of $\mu_2$, and hence of $\bar\mu_2$ in the case of continuous graphons, but it does not provide a rate of convergence that would allow us to estimate $\bar\mu_2$ for a given $n$. Nevertheless, similar to Proposition~\ref{prop:SBM}, it specifies the importance of the infimum of the degree function $d(x)$ (minimum degree of $D_{\SBM}$), as a potential value for $\bar\mu_2$.
In the case of SBM graphons, the range of the degree function is given by a set of measure zero, and that is why it can be isolated from potential eigenvalues as it has been shown by Proposition~\ref{prop:SBM}.

Finally, notice that the bound of Theorem~\ref{thm:arrivals} implies the computation of the spectrum of several matrices. However, with Proposition~\ref{prop:SBM}, this computation is reduced to the spectrum of matrices of size $m$ that depends only on the graphon and not on the number of agents, which is practical for open multi-agent systems over large graphs.

\section{Conclusions and future work}

In this paper, we showed that graphons can be used as a tool to analyze open multi-agent systems when the network topologies are sampled from a graphon. This required the analysis of the Laplacian spectrum of graphs sampled from graphons, which is not often done, since the theory of graphons is related to the adjacency matrix. 

We have some results highlighting the importance of graphons in open multi-agent systems, but there are important extensions in this direction.

It would be interesting to extend the results to the case of sampling using stochastic latent variables \cite{avella2018centrality}, which can be more appropriate for the case of open multi-agent systems. Also, it would be important to analyze the case of dependence on the sequence of graph topologies sampled during the events (arrivals, departures, replacements). Furthermore, the analysis of the stochastic case when the occurrence of the events is determined by a stochastic process (e.g., renewal processes) is an interesting area of research. Finally, it would be
important to extend the results to more general classes of graphons that cannot be decomposed as SBM graphons.

\bibliographystyle{IEEEtran}
\bibliography{arxiv}

\begin{thebibliography}{10}
\providecommand{\url}[1]{#1}
\csname url@samestyle\endcsname
\providecommand{\newblock}{\relax}
\providecommand{\bibinfo}[2]{#2}
\providecommand{\BIBentrySTDinterwordspacing}{\spaceskip=0pt\relax}
\providecommand{\BIBentryALTinterwordstretchfactor}{4}
\providecommand{\BIBentryALTinterwordspacing}{\spaceskip=\fontdimen2\font plus
\BIBentryALTinterwordstretchfactor\fontdimen3\font minus
  \fontdimen4\font\relax}
\providecommand{\BIBforeignlanguage}[2]{{%
\expandafter\ifx\csname l@#1\endcsname\relax
\typeout{** WARNING: IEEEtran.bst: No hyphenation pattern has been}%
\typeout{** loaded for the language `#1'. Using the pattern for}%
\typeout{** the default language instead.}%
\else
\language=\csname l@#1\endcsname
\fi
#2}}
\providecommand{\BIBdecl}{\relax}
\BIBdecl

\bibitem{hendrickx2016open}
J.~M. Hendrickx and S.~Martin, ``Open multi-agent systems: Gossiping with
  deterministic arrivals and departures,'' in \emph{2016 54th Annual Allerton
  Conference on Communication, Control, and Computing (Allerton)}, 2016, pp.
  1094--1101.

\bibitem{hendrickx2017open}
------, ``Open multi-agent systems: Gossiping with random arrivals and
  departures,'' in \emph{2017 IEEE 56th Annual Conference on Decision and
  Control (CDC)}.\hskip 1em plus 0.5em minus 0.4em\relax IEEE, 2017, pp.
  763--768.

\bibitem{vizuete2024trends}
R.~Vizuete, C.~Monnoyer~de Galland, P.~Frasca, E.~Panteley, and J.~M.
  Hendrickx, ``Trends and questions in open multi-agent systems,'' in
  \emph{Hybrid and Networked Dynamical Systems}.\hskip 1em plus 0.5em minus
  0.4em\relax Springer, 2024.

\bibitem{monnoyer2021random}
C.~Monnoyer~de Galland, R.~Vizuete, J.~M. Hendrickx, P.~Frasca, and
  E.~Panteley, ``Random coordinate descent algorithm for open multi-agent
  systems with complete topology and homogeneous agents,'' in \emph{2021 60th
  IEEE Conference on Decision and Control (CDC)}, 2021, pp. 1701--1708.

\bibitem{monnoyer2020open}
C.~Monnoyer~de Galland, S.~Martin, and J.~M. Hendrickx, ``Modelling gossip
  interactions in open multi-agent systems,'' \emph{arXiv preprint
  arXiv:2009.02970}, 2020.

\bibitem{vizuete2022resource}
R.~Vizuete, C.~M. de~Galland, J.~M. Hendrickx, P.~Frasca, and E.~Panteley,
  ``Resource allocation in open multi-agent systems: an online optimization
  analysis,'' in \emph{2022 IEEE 61st Conference on Decision and Control
  (CDC)}.\hskip 1em plus 0.5em minus 0.4em\relax IEEE, 2022, pp. 5185--5191.

\bibitem{franceschelli2020stability}
M.~Franceschelli and P.~Frasca, ``Stability of open multiagent systems and
  applications to dynamic consensus,'' \emph{IEEE Transactions on Automatic
  Control}, vol.~66, no.~5, pp. 2326--2331, 2020.

\bibitem{deplano2024stability}
D.~Deplano, M.~Franceschelli, and A.~Giua, ``Stability of paracontractive open
  multi-agent systems,'' in \emph{2024 IEEE 63rd Conference on Decision and
  Control (CDC)}, 2024, pp. 3031--3036.

\bibitem{deplano2025optimization}
D.~Deplano, N.~Bastianello, M.~Franceschelli, and K.~H. Johansson,
  ``Optimization and learning in open multi-agent systems,'' \emph{arXiv
  preprint arXiv:2501.16847}, 2025.

\bibitem{lovasz2006limits}
L.~Lov{\'a}sz and B.~Szegedy, ``Limits of dense graph sequences,''
  \emph{Journal of Combinatorial Theory, Series B}, vol.~96, no.~6, pp.
  933--957, 2006.

\bibitem{lovasz2012large}
L.~Lov{\'a}sz, \emph{Large networks and graph limits}.\hskip 1em plus 0.5em
  minus 0.4em\relax American Mathematical Soc., 2012, vol.~60.

\bibitem{gao2019graphon}
S.~Gao and P.~E. Caines, ``Graphon control of large-scale networks of linear
  systems,'' \emph{IEEE Transactions on Automatic Control}, vol.~65, no.~10,
  pp. 4090--4105, 2019.

\bibitem{parise2023graphon}
F.~Parise and A.~Ozdaglar, ``Graphon games: A statistical framework for network
  games and interventions,'' \emph{Econometrica}, vol.~91, no.~1, pp. 191--225,
  2023.

\bibitem{vizuete2020graphon}
R.~Vizuete, P.~Frasca, and F.~Garin, ``Graphon-based sensitivity analysis of
  {SIS} epidemics,'' \emph{IEEE Control Systems Letters}, vol.~4, no.~3, pp.
  542--547, 2020.

\bibitem{vizuete2021laplacian}
R.~Vizuete, F.~Garin, and P.~Frasca, ``The {L}aplacian spectrum of large graphs
  sampled from graphons,'' \emph{IEEE Transactions on Network Science and
  Engineering}, vol.~8, no.~2, pp. 1711--1721, 2021.

\bibitem{petit2021random}
J.~Petit, R.~Lambiotte, and T.~Carletti, ``Random walks on dense graphs and
  graphons,'' \emph{SIAM Journal on Applied Mathematics}, vol.~81, no.~6, pp.
  2323--2345, 2021.

\bibitem{bramburger2023pattern}
J.~Bramburger and M.~Holzer, ``Pattern formation in random networks using
  graphons,'' \emph{SIAM Journal on Mathematical Analysis}, vol.~55, no.~3, pp.
  2150--2185, 2023.

\bibitem{kuehn2019power}
C.~Kuehn and S.~Throm, ``Power network dynamics on graphons,'' \emph{SIAM
  Journal on Applied Mathematics}, vol.~79, no.~4, pp. 1271--1292, 2019.

\bibitem{bonnet2022consensus}
B.~Bonnet, N.~P. Duteil, and M.~Sigalotti, ``Consensus formation in first-order
  graphon models with time-varying topologies,'' \emph{Mathematical Models and
  Methods in Applied Sciences}, vol.~32, no.~11, pp. 2121--2188, 2022.

\bibitem{prisant2024opinion}
R.~Prisant, F.~Garin, and P.~Frasca, ``Opinion dynamics on signed graphs and
  graphons: Beyond the piece-wise constant case,'' in \emph{2024 IEEE 63rd
  Conference on Decision and Control (CDC)}, 2024, pp. 5430--5435.

\bibitem{vizuete2024SIS}
R.~Vizuete, P.~Frasca, and E.~Panteley, ``{SIS} epidemics on open networks: A
  replacement-based approximation,'' in \emph{2024 European Control Conference
  (ECC)}, 2024, pp. 1602--1608.

\bibitem{avella2018centrality}
M.~{Avella-Medina}, F.~{Parise}, M.~T. {Schaub}, and S.~{Segarra}, ``Centrality
  measures for graphons: Accounting for uncertainty in networks,'' \emph{IEEE
  Transactions on Network Science and Engineering}, vol.~7, no.~1, pp.
  520--537, 2020.

\bibitem{monnoyer2024random}
C.~Monnoyer~de Galland, R.~Vizuete, J.~M. Hendrickx, E.~Panteley, and
  P.~Frasca, ``Random coordinate descent for resource allocation in open
  multiagent systems,'' \emph{IEEE Transactions on Automatic Control}, vol.~69,
  no.~11, pp. 7600--7613, 2024.

\bibitem{carletti2008birth}
T.~Carletti, D.~Fanelli, A.~Guarino, F.~Bagnoli, and A.~Guazzini, ``Birth and
  death in a continuous opinion dynamics model: The consensus case,'' \emph{The
  European Physical Journal B}, vol.~64, pp. 285--292, 2008.

\bibitem{torok2013opinions}
J.~T{\"o}r{\"o}k, G.~I{\~n}iguez, T.~Yasseri, M.~San~Miguel, K.~Kaski, and
  J.~Kert{\'e}sz, ``Opinions, conflicts, and consensus: Modeling social
  dynamics in a collaborative environment,'' \emph{Physical Review Letters},
  vol. 110, no.~8, p. 088701, 2013.

\bibitem{gao2019spectral}
S.~Gao and P.~E. Caines, ``Spectral representations of graphons in very large
  network systems control,'' in \emph{2019 IEEE 58th conference on decision and
  Control (CDC)}.\hskip 1em plus 0.5em minus 0.4em\relax IEEE, 2019, pp.
  5068--5075.

\bibitem{vizuete2022contributions}
R.~Vizuete, ``Contributions to open multi-agent systems: consensus,
  optimization and epidemics,'' Ph.D. dissertation, Universit{\'e}
  Paris-Saclay, 2022.

\bibitem{varma2018open}
V.~S. Varma, I.-C. Mor{\u{a}}rescu, and D.~Ne{\v{s}}i{\'c}, ``Open multi-agent
  systems with discrete states and stochastic interactions,'' \emph{IEEE
  Control Systems Letters}, vol.~2, no.~3, pp. 375--380, 2018.

\bibitem{vizuete2020influence}
R.~Vizuete, P.~Frasca, and E.~Panteley, ``On the influence of noise in
  randomized consensus algorithms,'' \emph{IEEE Control Systems Letters},
  vol.~5, no.~3, pp. 1025--1030, 2020.

\bibitem{hsieh2021optimization}
Y.-G. Hsieh, F.~Iutzeler, J.~Malick, and P.~Mertikopoulos, ``Optimization in
  open networks via dual averaging,'' in \emph{2021 60th IEEE Conference on
  Decision and Control (CDC)}.\hskip 1em plus 0.5em minus 0.4em\relax IEEE,
  2021, pp. 514--520.

\bibitem{harary2014graphical}
F.~Harary and E.~M. Palmer, \emph{Graphical Enumeration}.\hskip 1em plus 0.5em
  minus 0.4em\relax Elsevier, 2014.

\bibitem{garin2024corrections}
F.~Garin, P.~Frasca, and R.~Vizuete, ``{Corrections to and improvements on
  results from" The Laplacian spectrum of large graphs sampled from
  graphons"},'' \emph{arXiv preprint arXiv:2407.14422}, 2024.

\bibitem{van2023graph}
P.~Van~Mieghem, \emph{{Graph Spectra for Complex Networks}}.\hskip 1em plus
  0.5em minus 0.4em\relax Cambridge university press, 2023.

\bibitem{chu1955bounds}
J.~T. Chu, ``On bounds for the normal integral,'' \emph{Biometrika}, vol.~42,
  no. 1/2, pp. 263--265, 1955.

\bibitem{karrer2011stochastic}
B.~Karrer and M.~E. Newman, ``Stochastic blockmodels and community structure in
  networks,'' \emph{Physical Review E—Statistical, Nonlinear, and Soft Matter
  Physics}, vol.~83, no.~1, p. 016107, 2011.

\bibitem{airoldi2013stochastic}
E.~M. Airoldi, T.~B. Costa, and S.~H. Chan, ``Stochastic blockmodel
  approximation of a graphon: Theory and consistent estimation,''
  \emph{Advances in Neural Information Processing Systems}, vol.~26, 2013.

\end{thebibliography}

\end{document}